\documentclass[prl, twocolumn, superscriptaddress]{revtex4-2}
\usepackage{xcolor}
\usepackage{amsmath,amssymb,bm,mathtools}
\usepackage{physics}
\usepackage{upgreek}
\usepackage{graphicx}
\usepackage{dcolumn}
\usepackage{dsfont}
\usepackage{nccmath}
\usepackage{braket}
\usepackage[colorlinks=true, allcolors=blue]{hyperref}
\usepackage{amsthm}
\usepackage{units}
\usepackage{mathrsfs}
\usepackage{hhline}
\usepackage{subfigure}

\newtheorem{theorem}{Theorem}
\newtheorem{corollary}{Corollary}[theorem]
\newtheorem{example}{Example}
\theoremstyle{definition}

\newtheorem{problem}{Problem}
\newtheorem{lemma}[theorem]{Lemma}

\begin{document}
\title{Universal transversal gates}
\author{Pragati Gupta}
\affiliation{Institute for Quantum Science and Technology, University of Calgary, Alberta T2N 1N4, Canada}
\author{Andrea Morello}
\affiliation{School of
Electrical Engineering and Telecommunications, UNSW Sydney, Sydney, NSW 2052, Australia}
\author{Barry C. Sanders\href{https://orcid.org/0000-0002-8326-8912}}
\affiliation{Institute for Quantum Science and Technology, University of Calgary, Alberta T2N 1N4, Canada}
\begin{abstract}
A long-standing challenge in quantum error correction is the infeasibility of universal transversal gates,
as shown by the Eastin-Knill theorem. 
We obtain a necessary and sufficient condition for a quantum code to have universal transversal gates and show that the Eastin-Knill
no-go result is a 
special case that does not hold for a general error model. 
We present a code construction using~$n$~$d$-dimensional systems that changes the logical error probability from a lower bound~$ \Omega(\nicefrac{1}{n\log d})$ to an upper bound~$\mathcal O (\nicefrac{1}{n d})$  and allows exact correction of both local and correlated errors.
Our universality condition determines the existence of a universal gate set for any quantum error-correcting code.
\end{abstract}
\maketitle

Recent advances in quantum error correction have paved the way towards building large-scale quantum devices that are resilient to noise~\cite{Bluvstein2024,Google23Nature,EDN+21Nature}.
The paradigm of \textit{transversal} gates, meaning gates that 
do not entangle different physical systems within a code,
has been key as transversal gates allow manipulation of encoded states without spreading errors~\cite{G09arxiv}, thus are essential for applications such as quantum memories~\cite{DKLA02JMathPhys,T15RMP}, computation~\cite{G09arxiv,G98PRA} and metrology~\cite{ZZPJ18NatCommun}. 
However, the famous Eastin-Knill theorem prohibits transversal gates from forming a universal set for quantum codes that protect against all local errors~\cite{EK09PRL}.  
As a consequence,  constructions such as magic-state factories~\cite{BK05PRA} are required that would consume most resources on an error-corrected quantum processor~\cite{OC17PRA}, presenting a major roadblock to achieving universality~\cite{CTV17Nature}. 
We show that the Eastin-Knill no-go result 
does not hold for every code; we determine the necessary and sufficient condition for achieving universal transversal gates with local error correction.

Several approaches have been proposed to achieve universal fault-tolerant gates without magic states.
These include code switching~\cite{ADP14PRL}, code drift~\cite{PR13PRL}, and pieceable fault-tolerance~\cite{YTC16PRX}. 
In these schemes, universal gates are implemented via a sequence of transformations of the codespace using non-transversal gates. However, a huge space-time overhead is required to achieve fault-tolerance~\cite{CTV17Nature} as these techniques rely on 
\textit{intermediate error correction} to prevent errors from spreading within a code block.

Motivated by the need for universal transversal gates,
an alternative approach is followed to allow approximate error correction~\cite{FNASP20PRX,KD21PRL,LZ23npjQuantInfo}.
This line of research explores the trade-off between \emph{continuous symmetries}, which are prohibited according to the Eastin-Knill theorem~\cite{EK09PRL}, and correction of subsystem erasure.  
However, the error rate for codes with continuous symmetries has a lower bound instead of an upper bound, implying no guarantee that errors are corrected~\cite{KD21PRL}.
For codes that allow universal transversal gates, an exponentially large system size is required to suppress the error rate~\cite{FNASP20PRX}. 

Here, we consider the possibility of universal transversal gates for local error-correcting codes
without restrictions on the noise model. 
\begin{problem}
\label{prob:transversal}
Does a quantum code exist such that it exactly corrects local errors and has a universal transversal gate set?
\end{problem}
\noindent Such codes, if they exist, would offer error suppression with only a polynomial scaling in system size, and achieve universality without any overhead. However, the challenge in constructing such codes is that conditions required to guarantee universality are unknown. 
Thus, we ask a general question and consider any universal gate set that can be implemented without magic states. 
\begin{problem}
\label{prob:universal}
Can universal fault-tolerant gates be constructed by using only unitary gates and intermediate error correction without employing ancillary states?
\end{problem}
\noindent
We answer the questions in both problems.

Our main technical contribution is a unifying condition that relates universality to the Lie algebra\"ic properties of  errors.
First we establish a minimum requirement for achieving both universality and fault-tolerance,
which applies to transversal as well as non-transversal gates. 
Specifically, universality can only be achieved when the set of correctable errors is \emph{not closed} under commutation relations.
We determine the necessary and sufficient condition using finite-dimensional representations of a code~\cite{KL97PRA}.
By extension of this general result, we find conditions for universal transversal gates and prove that such codes would correct both local and correlated errors.
We support this claim with examples of single- and multi-system codes that correct angular momentum shifts.

In contrast to prior works~\cite{FNASP20PRX,KD21PRL}, our construction for codes with universal transversal gates provides an upper bound on the error rate, instead of a lower bound,
implying a guaranteed limit to error.
Our approach enables a complete classification of code symmetries for a variety of applications, including codes for holographic anti-de Sitter/conformal field theory (AdS/CFT) correspondence with no global symmetry~\cite{h19arxiv}, which is impossible with either the Eastin-Knill theorem or the approximate Eastin-Knill theorem. 
We emphasize that our universality condition is also applicable to non-transversal gates, implying pieceable fault-tolerance~\cite{YTC16PRX}, but, without  using ancillary states.

\paragraph{General error model:}
We describe errors for an arbitrary noise model using a \emph{graded algebra}~\cite{KLV00PRL}. Let~$\mathcal E_1$ be a set of error generators comprising most-likely errors. For multi-qubit systems under a local-noise model,~$\mathcal E_1$ consists of single-qubit errors, and `weight' is the number of qubits on which an error acts non-trivially. For a general noise model, we  consider the set of higher-weight errors to be linear combinations of products of errors in~$\mathcal E_1$, i.e.,~$\mathcal E_t = \mathcal E_1^t$, where~$t\in \mathbb N$ is the error weight. Let~$\mathcal E_0:=\{\mathds 1\}$, the identity set meaning no error; the set of all errors is the graded algebra~$\mathcal E = \bigoplus_{t\in \mathbb N_0}\mathcal E_t$, where all~$\mathcal E_t$ are closed under the adjoint operation.
Here, we describe the action of~$\mathcal E$  on a Hilbert space~$\mathscr H$ as a set of operators~$\widehat{\mathcal E}\in \mathcal L(\mathscr H)$, where~$\mathcal L(\mathscr H)$ is the set of linear operators,  that satisfy the algebra\"ic properties of~$\mathcal E$.  

\paragraph{Main results:}
Let~$C:\mathscr H_K\to\mathscr H_N$ be a quantum code mapping a~$K$-dimensional logical Hilbert space to an~$N$-dimensional physical Hilbert space.
We refer to the span of all correctable errors and their commutators as the Lie algebra generated by errors,
and we refer to the set of generators  of unitary operators,
in the sense of the exponential map from algebra to group,
as
$\mathfrak{g}=\mathfrak u(N)$, as the physical algebra.
Note that the word `generate' is used in two ways but always clear from context.
The \emph{universality condition} is that the ``Lie algebra generated by errors forms the physical algebra''. 
\begin{theorem}\label{thm:universality}
    Given a quantum error-correcting code, the set of logical fault-tolerant operators is universal  iff the universality condition is satisfied.
\end{theorem}
\noindent 
For multi-qubit systems, universality can be achieved if the code corrects more than one erasure, enabling fault-tolerant non-transversal gates without ancillary qubits. Theorem~\ref{thm:universality} is most powerful for higher-dimensional systems that have enough redundancy to allow encoding information within a single physical system. In this case, the above condition enables universal gates without intermediate error correction. 

Next, we apply Thm.~\ref{thm:universality} to transversal gates. 
Consider multiple physical subsystems with~$\mathscr H_N = \mathscr H_{1}\otimes\cdots\otimes\mathscr H_{n}$, where~$n>1$ is the number of subsystems,~$\text{dim}(\mathscr H_i)=d_i$ and~$N=\prod_id_i$. 
We consider fault-tolerant operators of the form~$\hat G := \bigotimes_{i=1}^{n}\hat G_i$,   
where~$\hat G \in \mathcal U(\mathscr H_N)$,~$\hat G_i\in \mathcal U(\mathscr H_i)$ for~$i\in\{1,\dots,n\}$, and~$\mathcal U(\mathscr{H})$ is a set of unitary operators. 
The \emph{continuity condition} is that the ``errors do not form a Lie algebra''.
\begin{corollary}\label{cor:continuous}
    A quantum error-detecting code has a continuous group of logical transversal gates  only if  errors on each subsystem satisfy the continuity condition.
\end{corollary}
\noindent The set of transversal gates cannot be universal if it does not form a continuous group~\cite{FNASP20PRX}; thus, satisfying the continuity condition for local errors is necessary to achieve universal transversal gates.

The Eastin-Knill theorem,
which states that the set of logical transversal gates is not universal if a code can detect \textit{all} local errors,
is a special case of Cor.~\ref{cor:continuous}
as we now show.
If a code can detect all local errors,
then
the set of detectable errors~$\mathcal E_1$ is 
\begin{equation}
    \mathcal E_1 = \bigoplus_{i=1}^n\mathfrak{u}(d_i),
\end{equation}
which is a semisimple Lie algebra expressed as a direct sum of simple Lie algebras over all~$n$ subsystems.
Thus, Cor.~\ref{cor:continuous}
requires that such a quantum error-detecting code must not have a continuous group of logical transversal gates. The continuity condition for local errors is necessary for universal transversal gates, implying the Eastin-Knill no-go result.

Universal transversal gates are achievable when each physical subsystem is large enough to \emph{encode} the logical information using a finite-dimensional code. 
\begin{corollary}\label{cor:transversal}
    The set of logical transversal gates is universal  iff each subsystem has a non-trivial quantum error-correcting code that satisfies the universality condition.
\end{corollary}
\noindent 
In our construction, quantum codes have 1) multiple subsystems to store information non-locally, and 2) sufficiently large Hilbert space dimension for each subsystem to allow non-trivial encoding of the logical information. 
This `double redundancy'---i.e., employing redundancy of the logical encoding in a single-system code for correcting local errors, and then doubling this redundancy by inserting the single-system code into a multi-system code---enables universal transversal gates without 
sacrificing local error correction.
Most quantum codes, in contrast, only use entanglement between multiple subsystems to protect against local noise, i.e., are only singly redundant~\cite{KL97PRA}.

We call our construction \emph{cascading}---encoding logical information using multiple codes.
In contrast to concatenation, which aims at decreasing the logical error rate, cascading aims to provide an alternative form of protection from  subsystem erasure and correlated noise. 
Unlike cascaded codes, concatenated codes cannot achieve universal gates without intermediate error correction~\cite{YTC16PRX} and cannot correct strongly-correlated errors~\cite{AKP06PRL}.

Finally, we consider the error rate for codes with universal transversal gates. 
\begin{theorem}\label{thm:errorrate}
    Given a quantum error-correcting code with universal transversal gates, the error rate is~$\mathcal O(\nicefrac{1}{n\min d_i})$ for local noise and~$\mathcal O(\nicefrac{1}{\min d_i})$ for correlated noise.
\end{theorem}
\noindent 
Here, by correlated noise we mean errors on up to~$n$ subsystems, including local errors. The proof results from a counting argument that determines probability of an uncorrectable error, and 
representation theory, which imposes a fundamental bound on the allowed code distance.

\paragraph{Illustrative examples:} As an application of the above results, 
we 
construct codes with universal fault-tolerant gates and consider the encoding of a logical qudit in single or multiple  spin-$J$ systems having errors generated by arbitrary spin rotations, i.e.~$\mathcal E_1:=\text{span}\{\mathds 1,J_x,J_y,J_z\}^1$.
We say a spin code has distance~$D=(2t+1)$  if it can correct all errors in~$\mathcal E_t$.

We first discuss a simple example to illustrate the use of Thm.~\ref{thm:universality} for single-subsystem codes. We consider the smallest spin code with~$D=3$~\cite{G21PRL}. 
\begin{example}[Qubit in spin-$\nicefrac{7}{2}$ qudit]\label{ex:quoctit}
Spin codes with~$D=3$ have fault-tolerant logical Pauli gates.
\end{example}
\noindent 
The Lie algebra generated by correctable errors is~$\mathfrak L(\mathcal E_1)=\mathfrak u(2)$, whereas the physical algebra for spin-$\nicefrac{7}{2}$ is~$\mathfrak u(8)$; 
thus, universality cannot be achieved. The set of fault-tolerant gates is~$\{\widehat X_\text L/\widehat Y_\text L/\widehat Z_\text L:=e^{-i\pi J_{x,y,z}}
\}$, which are the logical Pauli gates. 

An important class of codes is~$D\geq$5 spin codes that correct errors in~$\mathcal E_2$. As~$\mathfrak L(\mathcal E_2) = \mathfrak u(2J+1)\forall J$. Thm.~\ref{thm:universality} asserts that the set of fault-tolerant gates must be universal; now we construct such a set. In comparison, existing ways of finding fault-tolerant gates, such as through finite groups~\cite{KT23PRL}, only lead to non-universal gates.
\begin{example}[Qubit in spin-$\nicefrac{25}{2}$ qudit]\label{ex:spin26}
    Spin codes with~$D\geq5$ have a universal set of fault-tolerant gates~$U_\text L:=\{\widehat P_\text L(\phi),  \widehat {SX}_\text L, \widehat {CZ}_\text L \}$.
\end{example}
\noindent 
Consider the  spin code with~$J=\nicefrac{25}{2}$ and codewords~\cite{LLA23PRA}
\begin{align}
\label{eq:codewords26}
    \ket{0}_\text L
    =&\sqrt{\nicefrac{1}{16}}\ket{-\nicefrac{25}{2}}+
    \sqrt{\nicefrac{10}{16}}\ket{-\nicefrac{5}{2}}+
    \sqrt{\nicefrac{5}{16}}\ket{+\nicefrac{15}{2}},\nonumber\\
    \ket{1}_\text L=&\sqrt{\nicefrac{1}{16}}\ket{+\nicefrac{25}{2}}+
    \sqrt{\nicefrac{10}{16}}\ket{+\nicefrac{5}{2}}+
    \sqrt{\nicefrac{5}{16}}\ket{-\nicefrac{15}{2}},
\end{align}
where~$\ket{m}$ denotes an eigenstate in the~$\hat J_z$ basis and the magnetic quantum number~$m\in\{-J,-J+1,\ldots,J\}$. 
The support of the codewords after an error is given by the sets~$m_{0/1} := \{m| \bra{m}\hat E\ket{0/1}_\text L\neq 0 , \hat E\in \widehat{\mathcal E_2} \}$, which are disjoint as any~$\hat E\in \widehat{\mathcal E_2}$ changes~$m$ by at most~$\pm 2$.
Let the logical Phase gate be
\begin{equation}\label{eq:spinphase}
   \widehat P_{\text{L}}(\phi):= \sum_{m\notin m_1}\ketbra m  +e^{i\phi}\sum_{m \in m_1}\ketbra m,
\end{equation}
for any~$\phi \in \mathbb R$.
Similarly, let the logical two-qubit gate be
\begin{equation}\label{eq:spincz}
    \widehat{CZ}_\text L := \sum_{m\notin m_1} \ketbra m\otimes \mathds 1 + \sum_{m\in m_1} \ketbra m\otimes e^{-i\pi\hat J_z}.
\end{equation}
The constructions  of~$\widehat P_{\text{L}}(\phi)$ and~$\widehat{CZ}_\text L$ are fault-tolerant by definition, and can be implemented through phase shifts using the so-called virtual gates~\cite{MWSCG17PRA} and geometric phase gates, respectively. We complete a universal set by constructing the logical~$\widehat{SX}_\text L$ gate, the square-root of the~$NOT$ gate, as 
\begin{equation}\label{eq:spinsx}
    \widehat {SX}_\text L:=e^{-i\frac{\pi}{2}(\hat J_x+\hat J_x^2) },
\end{equation}
where~$e^{-i\frac{\pi}{2}(\hat J_x+\hat J_x^2) }\ket{m} = \frac{\ket{m}+i\ket{-m}}{\sqrt{2}}$, and the~$m\leftrightarrow-m$ symmetry of the codewords leads to the correct gate action. 
Equations~\eqref{eq:spinphase}-\eqref{eq:spinsx} provide a universal gate set for all~$D\geq5$ spin codes, up to virtual phase updates.

The smallest spin code (Ex.~\ref{ex:quoctit}) is practically feasible but does not lead to universal gates. On the other hand,~$D=5$ spin codes (Ex.~\ref{ex:spin26}) achieve universality but require a large-spin system. Now we show that universality can be achieved on near-term devices using~$J\geq \nicefrac{5}{2}$ Schr\"odinger spin-cat codes~\cite{OBGDM24PRXQuantum}, which correct biased noise.
It was previously unknown that the spin-cat code can achieve universality without employing magic states~\cite{OBGDM24PRXQuantum}. 
\begin{example}[Qubit in Schr\"odinger spin-cat code]
    Spin-cat codes with half integer spin-$J\geq \nicefrac{5}{2}$ have a universal set of bias-preserving gates~$\{\widehat{P}_\text L(\phi), \widehat{SX}_\text L, \widehat{CZ}_\text L \}$.
\end{example} 
\noindent We apply Thm.~\ref{thm:universality} by considering the Lie algebra generated by correctable errors, say~$\{\mathds 1, J_x, \ldots J_x^t\}$ for~$t=J-\nicefrac{1}{2}$, and the less dominant errors, say~$\{J_z\}$, i.e.~$\mathfrak L(\mathds 1, J_z, J_x, \ldots J_x^t) = \mathfrak u(2J+1)$~$\forall J\geq\nicefrac{5}{2}$; thus, we can construct a universal set using Eqs.~\eqref{eq:spinphase}-\eqref{eq:spinsx}.

\begin{figure}[b]
    \centering
    \includegraphics[width=\linewidth]{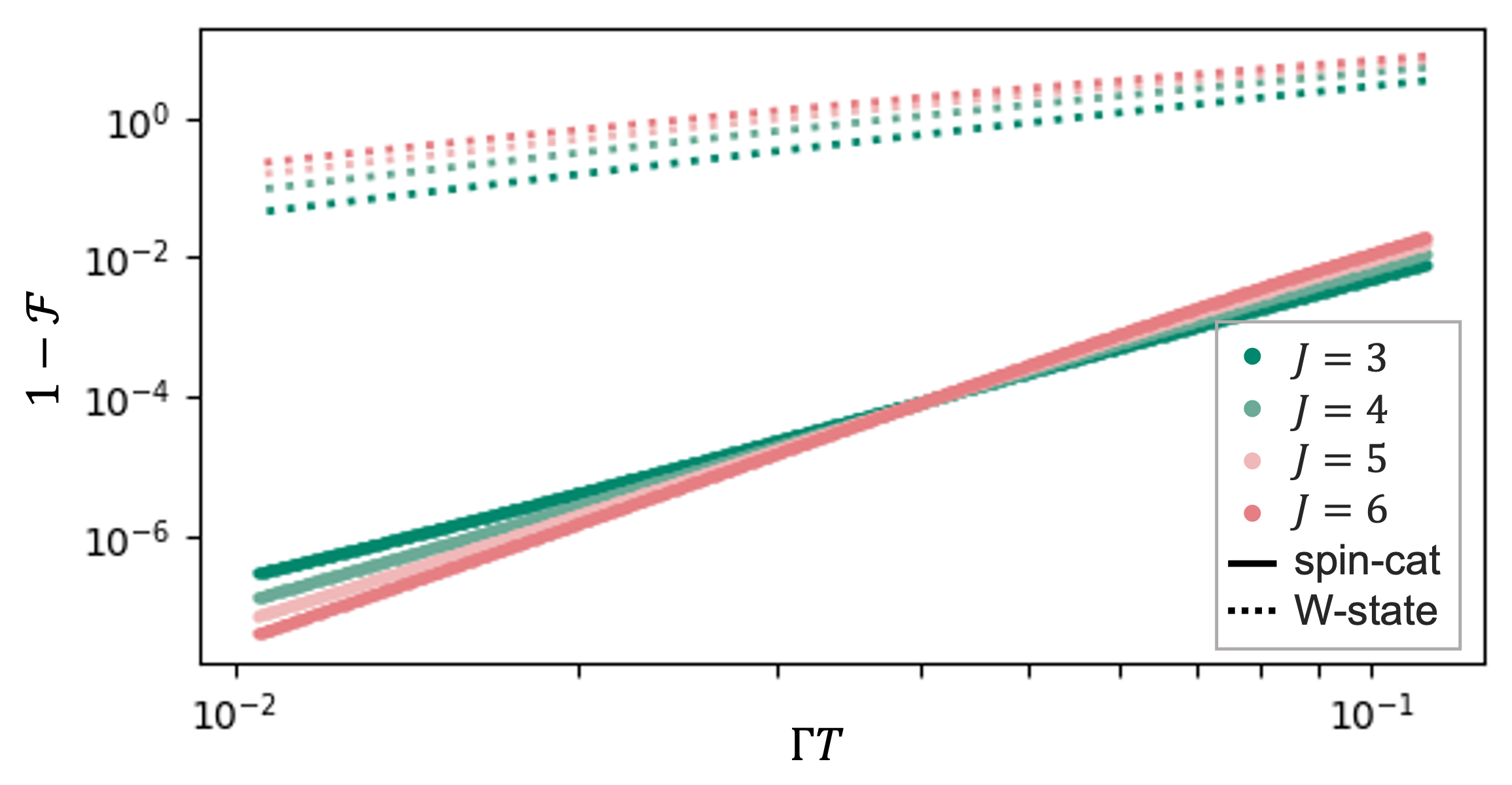}
    \caption{The multi spin-cat code has a finite code threshold given by the intersection of solid lines. For the W-state code, the infidelity~$1-\mathcal F$ increases with dimension-$(2J+1)$. Here,~$T$ is the evolution time and $\Gamma$ is the rate for $\hat J_x$ errors.}
    \label{fig:Wcat}
\end{figure}

Now we construct multi-subsystem codes with universal transversal gates. 
\begin{example}[Multi-spin codes]
    Multi-subsystem codes comprising~$D\geq5$ spin codes on each subsystem have universal transversal gates and can correct both local and correlated errors. 
\end{example}
\noindent 
To compare previous constructions for codes with universal transversal gates, such as the W-state code~\cite{FNASP20PRX}, we consider a similar multi-spin cat code with codewords
\begin{equation}
\ket {0/1}_\text L:= \frac{1}{\sqrt{n}}\left( \ket{\pm J 0\cdots0}+\cdots+ \ket{0\cdots0\pm J}\right),
\end{equation}
where the W-state code corresponds to $\ket{\pm 1}$ instead of $\ket{\pm J}$. Our numerical simulations, shown in Fig.~\ref{fig:Wcat}, confirm that the multi spin-cat code has a finite threshold, whereas the error rate for the W-state code increases with the subsystem size, performing worse than predicted by existing lower bounds~\cite{FNASP20PRX,KD21PRL}.

\begin{figure}
    \centering
    \includegraphics[width=\linewidth]{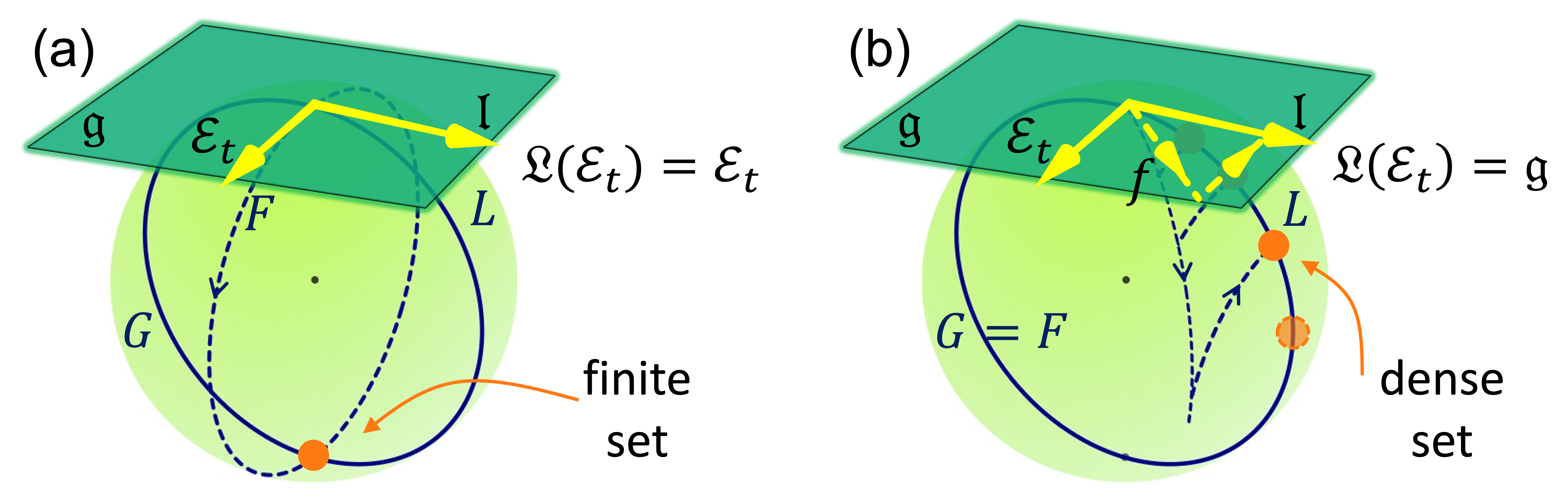}
    \caption{%
    The dark green plane represents the algebra~$\mathfrak g$,
    which forms the tangent space at some point on the pale green sphere, representing the group~$G$, which is a compact manifold.
    The `great circle'~$L$
    represents the set of logical operators,
    and the dashed loci~$F$ are the fault-tolerant gates,
    The solid yellow vectors are subspaces of the tangent space corresponding to correctable errors~$\epsilon$
    and generators~$\mathfrak l$ of~$L$. Logical fault-tolerant gates are the orange dots denoting points of intersection of~$L$ and~$F$. 
    (a)~The submanifolds~$L$ and~$F$ intersect at isolated points when~$\epsilon$ forms a Lie algebra; i.e.,~$\mathfrak L(\epsilon)= \epsilon$.
    (b)~The intersection of~$L$ and~$F$ forms a dense subset of~$L$ iff~$\mathfrak L(\epsilon)=\mathfrak g$, so,~$L\subseteq G=F$. 
}\label{fig:geometry}
\end{figure}

\paragraph{Idea for the proofs:} 
We present a Lie algebra\"ic approach to quantum error correction and consider the tangent space~$\mathfrak g$ at any point on the manifold formed by~$G$, the Lie group of unitary operators, as shown in Fig.~\ref{fig:geometry}. 
The space of generators~$\epsilon\subseteq \mathfrak g$ of correctable errors and~$\mathfrak l$ of unitary logical operators~$L$ are subspaces of~$\mathfrak g$,
depicted as one-dimensional subspaces in Fig.~\ref{fig:geometry}. 
The Lie algebra~$\mathfrak L(\epsilon)$ is the closure of~$\epsilon$ under the Lie bracket of~$\mathfrak g$.
We refer to~$\mathfrak l$ as the logical algebra of a quantum code.
Let~$F$ be the set of unitary fault-tolerant operators obtained by composing operators in~$\exp(\epsilon)$. Then the set of logical fault-tolerant operators~$L_F := F\cap L$ is a dense subset of~$L$ iff the action of~$f\in\mathfrak L(\epsilon)$ can be \emph{parallel} to~$\mathfrak l$. We formalize this intuitive geometric argument using Lie algebras and their representation.

The key ingredient in our proofs is a fundamental relation for quantum error correction that implies that the set of correctable errors and the logical algebra are `enough' to form the physical algebra for any finite-dimensional quantum code. This relation is 
mathematically stated as
\begin{equation}\label{eq:algebraicqec}
    \mathfrak g\sim \mathfrak e\otimes\mathfrak l\sim \epsilon \otimes_{\mathfrak l} \mathfrak l,
\end{equation}
where~$\mathfrak e$ is the set of correctable errors  that commutes with all logical generators. The tensor product algebra satisfies the commutation relations  of~$\mathfrak g$, and for~$E=e\otimes l$,~$E\otimes_{\mathfrak l} l' = e\otimes ll'\forall l'\in\mathfrak l$, where~$e\in \mathfrak e$ and~$l\in \mathfrak l$. 
We show that~$\mathfrak e$ and~$\mathfrak l$ are Lie algebras, and, for a representation of~$\mathfrak g$ on~$\mathscr H_N\sim\mathscr H_\text E\otimes\mathscr H_\text C
$~\cite{KL97PRA},~$\mathfrak e$  has support only on~$\mathscr H_\text E$ and~$\mathfrak l$ has support only on~$\mathscr H_\text C$. 
Thus, correctable errors that do not commute with all logical generators have a non-trivial 
representation on both~$\mathscr H_\text E$ and~$\mathscr H_\text C$ and generate fault-tolerant logical operators under the exponential map. 
If all correctable errors commute with logical generators, there are no  fault-tolerant gates.

Now we outline our argument for proving Thm.~\ref{thm:universality}.
We show that conditions for exact error correction are satisfied only when~$\mathfrak l \cap \epsilon = \mathds 1$, 
implying that $f\in\epsilon$ cannot be parallel to~$l\in\mathfrak l$.
Thus, 
if~$\mathfrak L(\epsilon)=\epsilon$, fault-tolerant operators
cannot  arbitrarily approximate logical operators close to identity, as illustrated in Fig.~\ref{fig:geometry}(a). 
The continuity condition, mathematically stated as~$\mathfrak L(\epsilon)\supsetneq \epsilon$, is required to ensure that $L_\text F$ does not form a discrete group.
For universality, the necessary and sufficient condition is~$\mathfrak l\subseteq \mathfrak L(\epsilon)$, which we show is equivalent to the universality condition~$\mathfrak L(\epsilon)=\mathfrak g$ using Eq.~\eqref{eq:algebraicqec}. Thus, satisfying the universality condition guarantees that~$L_F$ forms a dense subset of~$L$, as illustrated in Fig.~\ref{fig:geometry}(b).
For proving Cors.~\ref{cor:continuous} and~\ref{cor:transversal}, transversal gates may be viewed as local fault-tolerant operators that form a universal set if  the continuity and universality conditions are both satisfied for local errors on each subsystem.

For proving Thm.~\ref{thm:errorrate}, we estimate the probability of a logical error for some graded algebra using an \emph{algebra\"ic singleton bound} 
\begin{equation}\label{eq:singleton}
    |\mathfrak g| \geq |\mathcal E_t|^2|\mathfrak l|,
\end{equation}
where~$|\cdot|$ denotes cardinality. Here, we used Eq.~\eqref{eq:algebraicqec} and the fact that~$|\mathfrak e| 
= \operatorname{dim}(\mathscr H_\text E)^2$ with~$\operatorname{dim}(\mathscr H_\text E)\geq |\mathcal E_t|$.
Unlike the quantum singleton bound based on the no-cloning theorem~\cite{KL97PRA}, Eq.~\eqref{eq:singleton} holds  regardless of the tensor product structure of the physical Hilbert space or the type of noise.
See End Matter for detailed proofs of the main propositions.

\paragraph{Discussion:}
Our work inspires theoretical explorations based on the Lie algebra\"ic formulation of quantum error correction. One example is global symmetries in AdS/CFT correspondence, which requires that holographic codes must not have any transversal gates~\cite{h19arxiv}, contradicting the Eastin-Knill theorem that guarantees a discrete set of transversal gates~\cite{EK09PRL}. This contradiction is resolved by using our error model and coding within a physical subsystem: 
we know that a code has no transversal gates iff~$\epsilon=\mathfrak e$  for each subsystem, implying the logical algebra is an ideal of the physical algebra.
Other possible applications include topological codes~\cite{BK13PRL} and error correction for metrology~\cite{ZZPJ18NatCommun}  where local operators play an important role.

We anticipate that our results will be useful in near-term experiments for universal computation with fault-tolerance against hardware-level noise. Our code construction also enables correction of correlated noise, which has been a major source of errors for many platforms~\cite{WAK+21Nature}. Current platforms with high-dimensional control, such as  superconducting circuits~\cite{CWPB24arxiv,RSW+24arxiv}, neutral atoms~\cite{KYZ+24arxiv} and semiconducting devices~\cite{yu24arxiv} are well-suited for implementing universal computation using spin-cat codes.

\paragraph{Summary:}
Our universality condition is necessary and sufficient  for any quantum error-correcting code to  have universal fault-tolerant gates.
We have shown that the restriction of the Eastin-Knill theorem arises due to the set of local errors being closed under commutation relations, and codes with universal transversal gates can be constructed by applying our universality condition to single-subsystem codes. 
For a general error model, we have derived an upper bound~$\mathcal O(\nicefrac{1}{n\min_id_i})$ on the probability of a logical error  for codes with universal transversal gates, rather than a lower bound~$\Omega(\nicefrac{1}{n\max_i\log d_i})$ known from the approximate Eastin-Knill theorem. 
In our construction, both local errors and correlated errors can be exactly corrected, achieving exponential suppression of errors with a favorable scaling in system size.

\paragraph{Acknowledgements:}
P.G. and B.C.S. acknowledge support from Alberta Innovates and NSERC. A.M. acknowledges support of an Australian Research Council Laureate Fellowship (project no. FL240100181). 
\bibliography{main}
\section{End Matter}
Now we present the formal proofs for the main results. We begin by showing that~$\mathfrak l$  satisfies the following proposition.
\begin{lemma}\label{lem:logicalLie}
    The set of unitary logical operators is generated by a Lie algebra.
\end{lemma}
\begin{proof}
    As the set of unitary logical operators~$L$ is a closed subgroup of~$G$~\cite{EK09PRL}, 
    then,
    using Cartan's theorem~\cite{R02Lie},~$L$ is a Lie group generated by
    \begin{equation}\label{eq:logicalalgebra}
        \mathfrak l := \{g\in\mathfrak g|\exp (-itg) \in L\, \forall t\in \mathbb R \},
    \end{equation}
    which forms a Lie subalgebra~$\mathfrak l\subseteq \mathfrak g$~\cite{R02Lie}.
\end{proof}
\noindent 
Due to an extension of Cartan's theorem~\cite{R02Lie}, any dense subset of~$L$ must also be generated by the logical algebra.

Now we present prove the relation between~$\epsilon$,~$\mathfrak l$  and~$\mathfrak g$ presented in Eq.~\eqref{eq:algebraicqec}. 
\begin{lemma}\label{lem:singleton}
    For any finite-dimensional quantum error-correcting code, the composition of the set of correctable errors and the logical algebra forms the physical algebra.
\end{lemma}
\begin{proof}
Given~$C:\mathscr H_K\to\mathscr H_N$ with finite~$N$, 
its representation induces an isomorphism~$\mathscr H_N\sim\mathscr H_\text E\otimes\mathscr H_\text C$~\cite{KLV00PRL}. Let~$\mathfrak e$ be the Lie algebra with representation on~$\mathscr H_\text E$ and~$\mathfrak l$ is the Lie algebra with representation on~$\mathscr H_\text C$. 
Now, using  the induced isomorphism of the  associative algebra~$\mathcal L(\mathscr H_N)\sim\mathcal L(\mathscr H_\text E)\otimes\mathcal L(\mathscr H_\text C)$, and that the subspace of self-adjoint operators of the universal enveloping algebra forms the corresponding Lie algebra,
we get~$\mathfrak g \sim \mathfrak e \otimes \mathfrak l$.
By definition, correctable errors act on~$\mathscr H_\text E$~\cite{KL97PRA} so~$\mathfrak e\subseteq \epsilon$. Any~$E\in\epsilon$ can be expressed as~$E=e\otimes l$, for some~$e\in\mathfrak e$,~$l\in \mathfrak l$, and $\otimes_\mathfrak l$ satisfies~$E\otimes_{\mathfrak l} l' = e\otimes l l'\,\forall l'\in\mathfrak l$.
Thus,~$\mathfrak g\sim \mathfrak e\otimes\mathfrak l\sim \epsilon \otimes_{\mathfrak l}\mathfrak l$.
\end{proof}
\noindent
Let~$\mathscr H_\text E = \mathscr H_\text {E0}\oplus \mathscr H_\text {E1}$, where~$\mathscr H_\text {E0}$ is the space of no error and $\mathscr H_\text {E1}$ is the space of non-trivial errors, and~$\mathscr H_N\sim \left(\mathscr H_\text {E0}\oplus \mathscr H_\text {E1}\right)\otimes \mathscr H_\text C$. Then, operators in~$\mathcal L(\mathscr H_\text {E1})\otimes \mathcal L(\mathscr H_\text C)$ are correctable errors as they act trivially on the codespace~$\mathscr H_\text {E0}\otimes \mathscr H_\text C$.
So a correctable error~$E \in \epsilon$ can also be of the form~$E = e\otimes l$ for some non-trivial~$l\in \mathfrak l$. 

Now we address Prob.~\ref{prob:universal}.
We consider  fault-tolerant operators of the form~$\hat F = \overline{EC}_s \hat F_s  \cdots \overline{EC}_1 \hat F_1$,
where~$\overline {EC}_r$ are intermediate error correction steps with~$r\in\{1,\ldots,s\}$~\cite{YTC16PRX}.
We restrict~$\hat F_r\in \exp (\epsilon)$ to ensure that intermediate errors have at most weight-$t$ up to first order.
Using the Baker-Campbell-Hausdorff formula,~$F$ forms a Lie group generated by the Lie algebra~$\mathfrak f :=\mathfrak L(\epsilon)$. 
\begin{proof}[Proof of Thm.~\ref{thm:universality}]
The statement of the theorem is expressed mathematically as~$\mathfrak L(\epsilon)= \mathfrak g\Leftrightarrow L_F$ is a universal set.
We know that~$L_F$ forms a Lie subgroup of~$L$~\cite{R02Lie}, so
operators in~$L_F$ can only finitely approximate 
an element~$\hat L\in L\setminus L_F$ as~$L_F$ is closed under composition implying that the smallest distance between~$\hat L$ and~$L_F$ cannot be infinitesimal.
In other words, using proof by contradiction,
to form a universal set,
\begin{equation}\label{eq:universality}
    L_F = L\Leftrightarrow \mathfrak l\subseteq\mathfrak f
\end{equation}
must hold.
Using~$\mathfrak f=\mathfrak L(\epsilon)$, 
the `if' condition~$\mathfrak L(\epsilon)= \mathfrak g 
\implies \mathfrak l \subseteq\mathfrak L(\epsilon)$ is true because~$\mathfrak l\subseteq \mathfrak g$ by definition in Eq.~\eqref{eq:logicalalgebra} in Lem.~\ref{lem:logicalLie}.  
To prove the `only if' condition, we first note that~$\mathfrak L (\epsilon)\subseteq \mathfrak e$ iff~$\epsilon\subseteq \mathfrak e$ (using Lem.~\ref{lem:singleton}). 
Hence,~$\mathfrak l \subseteq\mathfrak L(\epsilon)$ implies~$\exists E\in\epsilon$ of the form~$E \sim e\otimes l$ with non-trivial~$l\in\mathfrak l$, and one such term is enough to generate~$\epsilon \otimes_{\mathfrak l} \mathfrak l$ via commutation relations. Thus,~$\mathfrak l \subseteq \mathfrak L(\epsilon)\implies \epsilon \otimes_{\mathfrak l} \mathfrak l\sim\mathfrak L(\epsilon)= \mathfrak g$ (using Lem.~\ref{lem:singleton}). 
\end{proof}

Now we address Prob.~\ref{prob:transversal}. We consider transversal gates generated by~$\mathfrak f_\text T:=\bigoplus_{i=1}^n\mathfrak L(\epsilon^{(i)})$, where~$\epsilon^{(i)}$ are errors on~$\mathscr H_i$.
We denote the projector from~$\mathscr H_N\to\mathscr H_\text C$ as~$\hat P_\text C$. 
\begin{proof}[Proof of Cor.~\ref{cor:continuous}]
    The corollary can be mathematically stated as~$\mathfrak L(\epsilon^{(i)}) = \epsilon^{(i)} \implies \mathfrak l\cap \mathfrak f_\text T=\{\mathds 1\}$.
    Consider a logical generator close to the origin~$ l\in \mathfrak l\cap\epsilon$ and its corresponding logical operator close to identity~$\hat L \in L$~\cite{R02Lie}.     As~$l\in\epsilon$,~$\hat P_\text C \hat L \hat P_\text C\propto \hat P_\text C$, and as~$l\in\mathfrak l$,~$\hat P_\text C\hat L\hat P_\text C = \hat L \hat P_\text C=\mathds 1\hat P_\text C$. Thus,~$\mathfrak l\cap \epsilon =\{\mathds 1\}$,
    and~$\mathfrak L(\epsilon)=\epsilon \implies\mathfrak f\cap\mathfrak l=\mathds 1$.
\end{proof}

\begin{proof}[Proof of Cor.~\ref{cor:transversal}]
    The theorem states that the set of transversal gates is universal iff each subsystem has a non-trivial
    representation of the code that satisfies the universality condition.
    Using Eq.~\eqref{eq:universality}, the set of transversal gates is universal iff~$\mathfrak l\subseteq\mathfrak f_\text T\subseteq \bigoplus_{i=1}^n \mathfrak u(d_i)$; i.e., the logical algebra acts locally, implying that the code is represented on each subsystem.
\end{proof}
\noindent

\begin{proof}[Proof of Thm.~\ref{thm:errorrate}]
    Let~$p$ be the maximum probability of an error in~$\mathcal E_1$. 
    For a local noise model, a logical error occurs when there are~$t+1$ or more errors  on the \emph{same} subsystem. 
    Now, the probability of a weight-$(t+1)$ error is~$p^{t+1}$ and
    the probability that~$(t+1)$ errors occur on a particular subsystem is~$\left(\nicefrac{1}{n}\right)^{t+1}$  
    for a local noise model. 
    We use Eq.~\eqref{eq:singleton} to relate~$t$ to the subsystem sizes:~$|\mathcal E_1^{(i)}|^{2(t+1)}\geq \nicefrac{d_i^2}{K^2}  \geq |\mathcal E_t^{(i)}|^2$, 
    giving~$t+1\geq \log \min_id_i$; 
    hence, the probability of a logical error~$p_{\text L} \leq n\left(\nicefrac{1}{n}\right)^{t+1}p^{t+1}\in\mathcal O(\nicefrac{1}{n\min_id_i})$. 
    It is important to note that, according to Thm.~\ref{thm:universality}, codes with  universal transversal gates must satisfy the universality condition, which is not possible with only local errors. 
    Thus, the set of correctable errors must include correlated errors, and, under correlated noise,~$p_\text L\leq p^{t+1}\in\mathcal O(\nicefrac{1}{\min_id_i})$.
\end{proof}

\end{document}